\documentclass[runningheads, envcountsame]{llncs}
\usepackage[mathlines]{lineno}
\usepackage{etoolbox}
\usepackage{url}
\usepackage{todonotes}
\usepackage{comment}
\usepackage[utf8]{inputenc}
\usepackage{rotating}
\usepackage{graphicx}
\usepackage{bm}

\usepackage{tikz}
\usetikzlibrary{arrows}
\usepackage{macros}

\newcommand*\linenomathpatch[1]{%
  \cspreto{#1}{\linenomath}%
  \cspreto{#1*}{\linenomath}%
  \csappto{end#1}{\endlinenomath}%
  \csappto{end#1*}{\endlinenomath}%
}
\newcommand*\linenomathpatchAMS[1]{%
  \cspreto{#1}{\linenomathAMS}%
  \cspreto{#1*}{\linenomathAMS}%
  \csappto{end#1}{\endlinenomath}%
  \csappto{end#1*}{\endlinenomath}%
}

\expandafter\ifx\linenomath\linenomathWithnumbers
  \let\linenomathAMS\linenomathWithnumbers
  \patchcmd\linenomathAMS{\advance\postdisplaypenalty\linenopenalty}{}{}{}
\else
  \let\linenomathAMS\linenomathNonumbers
\fi

\linenomathpatch{equation}
\linenomathpatchAMS{gather}
\linenomathpatchAMS{multline}
\linenomathpatchAMS{align}
\linenomathpatchAMS{alignat}
\linenomathpatchAMS{flalign}

\makeatletter
\patchcmd{\mmeasure@}{\measuring@true}{
  \measuring@true
  \ifnum-\linenopenaltypar>\interdisplaylinepenalty
    \advance\interdisplaylinepenalty-\linenopenalty
  \fi
  }{}{}
\makeatother

\begin{document}
\title{Computing Repairs Under Functional and Inclusion Dependencies via Argumentation
} 
%
\titlerunning{Computing Repairs via AFs}
%
\author{Yasir Mahmood 
\inst{1}
\and
Jonni Virtema
\inst{2}
\and 
Timon Barlag
\inst{3}
\and \\
Axel-Cyrille Ngonga Ngomo
\inst{1}}

\authorrunning{Mahmood et al.}
%
\institute{DICE group, Department of Computer Science, Paderborn University, Germany
		\email{yasir.mahmood@uni-paderborn.de}\\
		\email{axel.ngonga@upb.de} \\
	\and
	Department of Computer Science, University of Sheffield, United Kingdom\\
		\email{j.t.virtema@sheffield.ac.uk}\\
	\and 
	Institut f\"ur Theoretische Informatik, Leibniz Universit\"at Hannover, Germany
		\email{barlag@thi.uni-hannover.de}
	}
\maketitle              
\begin{abstract}
We discover a connection between finding subset-maximal repairs for sets of functional and inclusion dependencies,  and computing extensions within argumentation frameworks (AFs). 
We study the complexity of existence of a repair, and deciding whether a given tuple belongs to some (or every) repair, by simulating the instances of these problems via AFs.
We prove that subset-maximal repairs under 
functional dependencies correspond to the naive extensions, which also coincide with the preferred and stable extensions in the resulting AFs. 
For inclusion dependencies one needs a pre-processing step on the resulting AFs in order for the extensions to coincide.
Allowing both types of dependencies breaks this relationship between extensions and only preferred semantics captures the repairs.
Finally, we establish that the complexities of the above decision problems are $\NP$-complete and $\PiP$-complete, when both functional and inclusion dependencies are allowed.


\end{abstract}

\section{Introduction}

In real-world applications the provenance of data can be very diverse and include non-trustworthy sources. Thus databases are often inconsistent in practice due to the presence of integrity constraints, and a rich theory has been developed to deal with this inconsistency.
One of the main approaches for handling inconsistency is \emph{database repairing}. The goal is to identify and \emph{repair} inconsistencies in data in order to obtain a consistent database that satisfies the imposed constraints. In the usual approaches one would search for a database that satisfies the given constraints and differs minimally from the original database; the obtained database is called a \emph{repair} of the original.
Some of the most prominent notions of repairs are set-based repairs~\cite{tenCate:2012,barcelo2017data}, attribute-based repairs~\cite{Wijsen:2003}, and cardinality-based repairs~\cite{LopatenkoB07}.
Dung's abstract argumentation framework~\cite{Dung95a} has been specifically designed to model conflicts and support relationships among arguments. 
An abstract argumentation framework (AF) represents arguments and their conflicts through directed graphs and allows for a convenient exploration of the conflicts at an abstract level. 
AFs have been explored extensively 
for representation and reasoning with inconsistent knowledge-bases (KBs) covering datalog, existential rules, and description logics (see e.g., \cite{AriouaC16,AriouaCV17,AriouaTC15,YoungMR16,YunVC20,BienvenuB20} and \cite{ArieliBH19} for an overview). 
The common goal in each of these works is to formally establish a connection between inconsistent KBs and AFs such that the argumentation machinery then outputs extensions equivalent to the set of repairs of the KB. 
Nevertheless, in the setting of relational databases and integrity constraint, there is still a gap with respect to how or whether a connection between inconsistent databases and AFs can be established.
To the best of our knowledge, only functional dependencies (FDs) have been investigated in the context of AFs, as discussed in~\cite{BienvenuB20}.
We expand this area of research by establishing further connections between repairs and abstract argumentation frameworks when further integrity constraints are allowed. 

In this paper, we focus on subset repairs of relational databases when the integrity constraints are functional and inclusion dependencies (IDs).
We are interested in the computational problems of deciding the existence of a repair, and determining whether a given tuple belongs to some (or every) repair.
We show how subset-maximal repairs for a set of functional dependencies and inclusion dependencies can be obtained 
by computing the naive, preferred, or stable extensions (see Section~\ref{sec:preli} for definitions) in the related AFs.
Repairs under 
functional dependencies correspond to the naive extensions, which also coincide with the preferred and stable extensions in the resulting AFs. 
For inclusion dependencies one needs a pre-processing step on the resulting AFs in order for the extensions to coincide.
Allowing both types of dependencies breaks this relationship between extensions and only preferred semantics captures the maximal repairs.
Finally, we consider the complexity of deciding 
whether a tuple belongs to at least one or to all repairs, respectively. See Table \ref{table:cont} for the complexity results.

By employing Dung's argumentation framework to model repairs of a relational database, one can effectively abstract away from the detailed-content of individual entries in the database and focus solely on their relationships with other entries.
This approach provides a clearer understanding of why specific records either appear or do not appear in a repair, as well as the reasons certain values may be absent from query answers.
Furthermore, this modeling approach allows for the incorporation of additional information about records, such as priorities among them, directly at an abstract level. 

\paragraph{Related Work}
The problem of computing subset maximal repairs and its complexity has been explored extensively in the database setting~\cite{AfratiK09,ArenasBC01,CHOMICKI200590,HannulaW22,LivshitsKR20,StaworkoC10} (see \cite{Bertossi06,Bertossi19} for an overview).
The notion of conflict graphs and hypergraphs has been introduced before in the case of functional dependencies~\cite{kimelfeld-17,KimelfeldLP20,StaworkoCM12}.
In particular, a correspondence between repairs and subset maximal independent sets of the conflict graph for FDs has been established~\cite{ArenasBC01}. 
Notice that the same definition also yields a correspondence between repairs and the naive extensions when the conflict graph is seen as an argumentation framework. 
Nevertheless, up to our knowledge, no work has considered a similar graph representation when inclusion dependencies are taken into account.
%
Hannula and Wijsen~\cite{HannulaW22} addressed the problem of consistent query answering with respect to primary and foreign keys. 
Their setting allows the insertion of new tuples to fulfill foreign key constraints rather than only deleting.
Our work differs from the previous work, since it combines functional dependencies (a subclass of equality-generating dependencies) and inclusion dependencies (a subclass of tuple-generating dependencies).
Moreover, one of our main contributions lies in connecting repairs under FDs and IDs to the extensions of argumentation frameworks in Dung's setting~\cite{Dung95a}.
Finally, the connection between AFs and 
preferred repairs has been explored in the context of prioritized description logic~\cite{BienvenuB20} and datalog knowledge bases~\cite{AriouaTCB14,HoAASA22}.

\begin{table}[t]
	\centering
		\begin{tabular}{l@{\; }c@{\; }c@{\; }c@{\; }c@{\; }c@{\; }} 
			\toprule
			& Atoms & AF-semantics &  \multicolumn{3}{c}{Complexity Results} \\
		& &\multicolumn{1}{c}{for $\REP$} & $\REP$ & $\somerepair$ & $\allrepair$  \\ \midrule 

			& FDs & $\sigma \in \{\naive,\pref,\stab\}$ \scriptsize{(Thm. \ref{thm:ext-dep})} & (trivial) & (trivial) & $\in\Ptime$ \\
			& IDs & $\pref$ \scriptsize{(Thm. \ref{thm:ext-inc})} & $\in\Ptime$\textsuperscript{\cite{AfratiK09}} & $\in\Ptime$\textsuperscript{\cite{AfratiK09}} & $\in\Ptime$\textsuperscript{\cite{AfratiK09}} \\
			& FDs+IDs & $\pref$ \scriptsize{(Thm. \ref{thm:ext-both})} & $\NP$ \scriptsize{(Thm. \ref{thm:cred-both})} & $\NP$ \scriptsize{(Thm. \ref{thm:cred-both})} & $\PiP$  \scriptsize{(Thm. \ref{thm:skep-both})} \\
			
			\bottomrule
		\end{tabular}
	\caption{Overview of our main contributions. The complexity results depict completeness, unless specified otherwise. The second column indicates the AF-semantics corresponding to subset-repairs ($\REP$) for dependencies in the first column, and the later three columns present the complexity of each problem. The $\Ptime$-results are already known in the literature, whereas the remaining results are new.}\label{table:cont}
\end{table}
\section{Preliminaries}\label{sec:preli}
We assume that the reader is familiar with basics of complexity theory.
We will encounter, in particular, the complexity classes $\Ptime, \NP$, and $\PiP$.
In the following, we shortly recall the necessary definitions from databases and argumentation.


We begin by restricting our attention to unirelational databases as these suffice for establishing our desired connections to argumentation frameworks as well as to our hardness results (see Table \ref{table:cont}).
Towards the end (Sec.~\ref{sec:multi}), we highlight the required changes to expand this approach to the multirelational setting.
%
The unirelational case is also connected to the literature in team-semantics~\cite{vaananen07}, which is a logical framework where formulae are evaluated over  unirelational databases (teams in their terminology). 
In this setting the complexity of finding maximal satisfying subteams has been studied by Hannula and Hella~\cite{HannulaH22} for inclusion logic formulas and by Mahmood~\cite{Mthesis23} for propositional dependence logic. 
In the team-semantics literature, FDs are known as \emph{dependence} atoms and IDs as \emph{inclusion} atoms, denoted respectively as $\depa{x}{y}$ and $\inca{x}{y}$. \footnote{We borrow this notation and write $\depa{x}{y}$ and $\inca{x}{y}$ for FDs and IDs, respectively.}

\paragraph{Databases and Repairs}
For our setting, an instance of a database is a single table denoted as $T$. 
We call each entry in the table a \emph{tuple} which is associated with an identifier.
Formally, a table corresponds to a relational schema denoted as $T(x_1 ,\ldots, x_n)$, where $T$ is the relation name and $x_1,\ldots,x_n$ are distinct {attributes}.
For an attribute $x$ and a tuple $s\in T$, $s(x)$ denotes the value taken by $s$ for the attribute $x$ and for a sequence $\tuple x =(x_1,\ldots,x_k)$, $s(\tuple x)$ denotes the sequence of values $(s(x_1),\ldots, s(x_k))$.
Given a database $T$, then $\dom(T)$ denotes the \emph{active domain} of $T$, defined as the collection of all the values that occur in the tuples of $T$.

Let $T(x_1, \ldots, x_n)$ be a schema and $T$ be a database. 
A \emph{functional dependency} (FD) over $T$ is an expression of the form $\depa{x}{y}$ (also denoted as $\tuple x \rightarrow \tuple y$) for sequences $\tuple x, \tuple y$ of attributes in $T$.
A database $T$ satisfies $\depa{x}{y}$, denoted as $T\models \depa{x}{y}$ if for all $s,t\in T$: if $s(\tuple x)=t(\tuple x)$ then $s(\tuple y)=t(\tuple y)$.
That is, every two tuples from $T$ that agree on $\tuple x$ also agree on $\tuple y$.
Moreover, an \emph{inclusion dependency} (ID) is an expression of the form $\inca{x}{y}$ for two sequences $\tuple x$ and $\tuple y$ of attributes with same length.
The table $T$ satisfies $\inca{x}{y}$ ($T\models \inca{x}{y}$) if for each $s\in T$, there is some $t\in T$ such that $s(\tuple x)=t(\tuple y)$. 
Moreover, we call each such $t$ the \emph{satisfying tuple for} $s$ and $i$.
By a dependency atom, we mean either a functional or an inclusion dependency.

Let $T$ be a database and $B$ be a collection of dependency atoms.
Then $T$ is \emph{consistent} with respect to $B$, denoted as $T\models B$, if $T\models b$ for each $b\in B$.
Moreover, $T$ is \emph{inconsistent} with respect to $B$ if there is some $b\in B$  such that $T\not \models b$.
A \emph{subset-repair} of $T$ with respect to $B$ is a subset $P\subseteq T$ which is consistent with respect to $B$, and maximal in the sense that no set $P'$ exists such that it is consistent with respect to $B$ and $P \subset P' \subseteq T$.
In the following, we simply speak of a repair when we intend to mean a subset-repair.
Furthermore, we often consider a database $T$ without explicitly highlighting its schema.
Let $\calB= \langle T,B\rangle $ where $T$ is a database and $B$ is a set of dependency atoms, then $\repairs(\calB)$ denotes the set of all repairs for $\calB$.
Since an empty database satisfies each dependency trivially, we restrict the notion of a repair to non-empty databases.
The problem we are interested in ($\REP$) asks to decide whether there exists a repair 
for an instance $\calB$.
%
\problemdef{$\REP$}{an instance $\calB=\langle T, B\rangle$}{is is true that $\repairs(\calB)\neq \emptyset$}

Two further problems of interest are 
\emph{brave} and \emph{cautious} reasoning for a tuple $s\in T$. 
Given an instance $\calB =\langle T, B\rangle $ and tuple $s\in T$, then brave (cautious) reasoning for $s$ denoted as $\somerepair(s,\calB)$ ($\allrepair(s,\calB)$) asks whether $s$ belongs to some (every) repair for $\calB$.

\paragraph*{Abstract Argumentation}
We use Dung's argumentation framework~\cite{Dung95a} and consider only non-empty and finite sets of arguments~$A$.
An \emph{(argumentation) framework~(AF)} is a directed graph~$\calF=(A, R)$, where $A$ is a set of arguments and the relation $R \subseteq A\times A$ represents direct attacks between arguments.
If $S\subseteq A$, we say that an argument~$s \in A$ is \emph{defended by $S$ in $\calF$}, if for every $(s', s) \in R$ there exists $s'' \in S$ such that $(s'', s') \in R$.

In abstract argumentation one is interested in computing the so-called \emph{extensions}, which are subsets~$S \subseteq A$ of the arguments that have certain properties.
The set~$S$ of arguments is called \emph{conflict-free in~$\calF$} if $(S\times S) \cap R = \emptyset$.
Let $S$ be conflict-free, then $S$ is
\begin{enumerate}
	\item \emph{naive in $\calF$} if no $S' \supset S$ is \emph{conflict-free} in $\calF$;
	\item \emph{admissible in $\calF$} if every $s \in S$ is \emph{defended by $S$ in $\calF$}.
	
\end{enumerate}

Further, let 
$S$ be admissible. Then, $S$ is
\begin{enumerate}
	\setcounter{enumi}{2}
	\item \emph{preferred in~$\calF$}, if there is no $S' \supset S$ that is \emph{admissible in $\calF$};
	\item \emph{stable in~$\calF$} if every $s \in A \setminus S$ is \emph{attacked} by some $s' \in S$.
\end{enumerate}

We denote each of the mentioned semantics by abbreviations: $\conf, \naive,$ $\adm, \pref,$ and $\stab$, respectively.
For 
a semantics~$\sigma \in \{\conf, \naive, \adm, \pref,  \stab\}$, we write $\sigma(\calF)$ for the set of \emph{all extensions} of semantics~$\sigma$ in $\calF$ \footnote{We disallow the empty set ($\emptyset$) in extensions for the sake of compatibility with repairs. Nevertheless, one can allow $\emptyset$ as an extension in AFs and the empty database as repairs, without affecting our complexity results}
Now, we are ready to define the corresponding decision problem asking for extension existence with respect to a semantics $\sigma$.
\problemdef{$\sem\sigma $}{an argumentation framework~$\calF$}{is it true that $\sigma(\calF)\neq\emptyset$}

Finally, for an AF~$\calF{=}(A,R)$ and $a\in A$, the problems $\cred_{\sigma}$ and $\skep_{\sigma}$ ask whether~$a$ is in some $\sigma$-extension of $\calF$ (``\emph{credulously} accepted'') or every $\sigma$-extension of~$\calF$  (``\emph{skeptically} accepted''), respectively. 
The complexity of reasoning in argumentation is well understood, see~\cite[Table 1]{DvorakDunne17} for an overview.
In particular, $\cred_{\naive}$ and $\skep_{\naive}$ are in $\Ptime$, 
whereas, $\cred_{\pref}$  and $\skep_{\pref}$ are $\NP$-complete and $\PiP$-complete, respectively.
Moreover, the problem 
to decide whether there is a non-empty extension is in $\Ptime$ for $\naive$ and $\NP$-complete for $\pref$-semantics.
This makes $\naive$-semantics somewhat easier 
and $\pref$ the hardest among the considered semantics in this work.
\section{Inconsistent Databases via Argumentation Frameworks}\label{sec:rep-afs}
In the first two subsections, we consider instances containing only one type of dependemcy atoms to an AF. 
Then, we combine both (dependence and inclusion) atoms in the third subsection.
Given an instance $\calB = \langle T,B\rangle $ comprising a database $T$ and a set $B$ of dependencies, the goal is to capture all the subset-repairs for $\calB$ by $\sigma$-extensions of the resulting AF ($\mathcal F_\calB$) for some semantics $\sigma$.

In Section~\ref{sec:fd}, we encode an instance $\mathcal D = \langle T,D \rangle $ with database $T$ and a collection $D$ of FDs into an AF $\mathcal F_\calD$.
This is achieved by letting each tuple $s\in T$ be an argument.
Then the attack relation for $\AF{F}{D}$ simulates the violation between a pair $s,t\in T$ failing some $d\in  D$. 
Although the construction for FDs is similar to the approach adapted by Bienvenu and Bourgaux~\cite{BienvenuB20}, we do not consider priorities among tuples in the database and therefore establish that a weaker AF-semantics is already enough to capture repairs in our setting.

In Section~\ref{sec:inc}, we simulate 
an instance $\calI= \langle T, I\rangle$ including a collection $I$ of inclusion dependencies (IDs) via AFs.
The first observation is that the semantics of IDs requires the notion of \emph{support} or \emph{defense} rather than \emph{conflict} between tuples.
Then, we depict each tuple as an argument as well as use auxiliary arguments to simulate inclusion dependencies (i.e., to model the semantics for IDs).
This is achieved by letting $s_i$ be an argument for each $s\in T$ and $i\in I$ such that $s_i$ attacks $s$.
Then, the arguments defending $s$ against $s_i$ correspond precisely to the satisfying tuples $t\in T$ for $s$ and $i$.
Further, we add self-attacks for these auxiliary arguments to prohibit them from appearing in any extension.
Consequently, we establish a connection between repairs for $\calI$ and the extensions for AFs under preferred semantics. 
Finally, we establish that after a pre-processing on the resulting AF, the stable and naive extensions also yield repairs for $\calI$.

Having established that both FDs and IDs can be modeled in AFs via attacks, Section~\ref{sec:both} generalizes this approach by allowing both types of dependencies. 

\subsection{Simulating Functional Dependencies via AFs}\label{sec:fd}
We transform an instance $\mathcal D = \langle T,D \rangle $ with database $T$ and a collection $D$ of FDs to an AF $\mathcal F_\calD$ defined as follows.

\begin{definition}
	Let $\calD =\langle T, D\rangle$ be an instance of $\REP$ including a database $T$ and a collection $D$ of FDs. 
	Then, $\mathcal F_{\calD}$ denotes the following AF.
	\begin{itemize}
		\item  $A\dfn T$, that is, each $s\in T$ is seen as an argument.
		\item $R\dfn \{(s,t),(t,s) \mid \text{ there is some } d\in  D, \text{ s.t. } \{s,t\}\not\models d\}$.
	\end{itemize}
	We call $\AF{F}{D} $ the argumentation framework generated by the instance $\calD$. Moreover, we also call $R$ the \emph{conflict graph} for $\mathcal D$.
\end{definition}

Note that, for a given instance $\mathcal D$, the framework $\AF{F}{D}$ can be generated in polynomial time. 
The attack relation $R$ is constructed for each $d\in D$ by taking each pair $s,t\in T$ in turn and checking whether $\{s,t\}\models d$ or not.

\begin{example}\label{intro:ex-dep}
	Consider $\calD = \langle T,D \rangle$ with database $T=\{s,t,u,v\}$ as depicted inside table in Figure~\ref{fig:ex-dep} and FDs $\{\depas{\text{Emp\_ID}}{\text{Dept}}$, $ \depas{\text{Sup\_ID}}{\text{Building}}\}$. 
	Informally, each employee is 
	associated with a unique department and employees supervised by the same supervisor 
	work in the same building.
	Observe that, $\{s,t\}\not\models \depas{\text{Emp\_ID}}{\text{Dept}}$, $\{u,v\}\not\models \depas{\text{Emp\_ID}}{\text{Dept}}$, and $\{t,v\}\not\models \depas{\text{Sup\_ID}}{\text{Building}}$.
	The resulting AF $\AF{F}{D}$ is depicted on the right side  of  Figure~\ref{fig:ex-dep}.
	The preferred (as well as naive and stable) extensions of $\AF{F}{D}$ include $\{s,v\}, \{t,u\}$ and $\{s,u\}$. 
	Clearly, these three are the only repairs for $\calD$.

	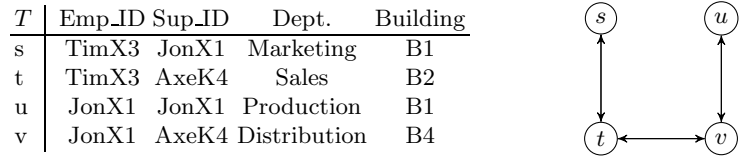
\begin{figure}[t]
		\centering
		\begin{tikzpicture}[scale=.8,arg/.style={circle,draw=black,fill=white,inner sep=.75mm}]
			\node[arg] (w) at (6,2) {$s$};
			\node[arg] (x) at (6,0) {$t$};
			\node[arg] (y) at (8,2) {$u$};
			\node[arg] (z) at (8,0) {$v$};
			
			\foreach \f/\t in {y/z,z/y,w/x,x/w,x/z,z/x}{
				\path[-stealth',draw=black] (\f) edge (\t);
			}
		
			\node (table) at (0,1) {
			\begin{tabular}{l @{\hskip5pt}| @{\hskip 5pt}c@{\;}ccc}
				$T$ & Emp\_ID & Sup\_ID & Dept.  & Building \\
				\hline 
				s & TimX3 & JonX1 & Marketing	  & B1  \\
				t & TimX3 & AxeK4 & Sales 		  & B2 \\
				u & JonX1 & JonX1 & Production 	  & B1 \\
				v & JonX1 & AxeK4 & Distribution  & B4 \\
			\end{tabular}};
		\end{tikzpicture}\\[-.75em]
		\caption{Argumentation framework for modelling FDs in Example~\ref{intro:ex-dep}.}\label{fig:ex-dep}
	\end{figure}
\end{example}

It is easy to observe that a subset $P\subseteq T$ satisfying each $d\in D$ contains precisely those tuples $s\in T$, which are not in conflict with each other.
Clearly, such subsets correspond to the naive extensions (maximal conflict-free sets) of $\AF{F}{D}$.
Moreover, since the attack relation in $\AF{F}{D}$ is symmetric, i.e., $(s,t)\in R$  iff $(t,s) \in R$, the preferred, stable and naive extensions coincide~\cite[Prop. 4 \& 5]{Coste-MarquisDM05}.

\begin{theorem}\label{thm:ext-dep}
	Let $\mathcal D = \langle T, D \rangle$ be an instance of $\REP$ where $D$ is a set of FDs and let $\AF{F}{D}$ denote the argumentation framework generated by $\mathcal D$.
	Then for every subset $P\subseteq T$, $P \in \repairs(\calD)$  iff $P\in \sigma(\AF{F}{D})$ for $\sigma \in \{\naive, \stab, \pref\}$.
\end{theorem}
\begin{proof}
	Let $\calD = \langle T, D \rangle$ be an instance of $\REP$ and $P\subseteq T$ such that $P\models d$ for each $d\in D$.
	Then, $P$ is clearly conflict-free in $\AF{F}{D}$.
	Moreover, since $P$ is a repair (and hence a maximal subset) of $T$, there is no $t\in T\setminus P$ such that $P\cup \{t\}$ is also conflict-free.
	As a result, $P$ is a naive extension in $\AF{F}{D}$.
	Finally, the same holds for preferred and stable extensions since the attack relation in $\AF{F}{D}$ is symmetric.

	Conversely, let $P\subseteq A$ be naive in $\AF{F}{D}$.
	Then, $\{s,t\}\models d$ for each $s,t\in P$ and $d\in D$ since $P$ is conflict-free.
	Moreover, $P$ is also subset maximal and therefore a repair for $\calD$. \qed
	
\end{proof}
An interesting corollary of Theorem~\ref{thm:ext-dep} reproves that a subset-repair for $\calD$ can be computed in polynomial time~\cite{DvorakDunne17}. 
Moreover we can also decide if a given tuple $s\in T$ is in some (or all) repairs, in polynomial time.
In fact, the basic properties of functional dependencies allow us to make the following observation regarding the acceptability of tuples with respect to $\calD$.
\begin{remark}\label{cor:dep-cs}
	Let $\mathcal D = \langle T, D \rangle$ be an instance of $\REP$ where $D$ is a set of FDs. 
	Then, $\somerepair(s,\calD)$ is true for  every $s\in T$, and $\allrepair(s,\calD)$ is true for a tuple $s\in T$ iff $\{s,t\} \models d$ for each $t\in T$ and $d\in D$.
\end{remark}

We conclude this section by observing that adding a size restriction for a repair renders the complexity of $\REP$ $\NP$-hard. Moreover, this already holds for propositional databases, that is, when $\dom(T)=\{0,1\}$.
The following result was proven in the context of team-semantics and maximal satisfying subteams for propositional dependence logic. 
\begin{theorem}\cite[Theorem 3.32]{Mthesis23}\label{thm:size-dep}
	There is an instance $\calD$ including a propositional database $T$ and FDs $D$, such that given $k\in \mathbb N$, the problem to decide whether there is a repair $P\subseteq T$ for $\calD$ such that $|P|\geq k$ is $\NP$-complete.
\end{theorem}

\subsection{Simulating Inclusion Dependencies via AFs}\label{sec:inc}

Let $\calI = \langle T,I\rangle $ be an instance of $\REP$ with a database $T$ and collection $I$ of IDs.
For $i \in  I$ (say $i= \inca{x}{y}$) and $s\in T$, let $t_1,\ldots, t_m\in T$ be such that $s(\tuple x) = t_j (\tuple y)$ for $j\leq m$.
Then we say that, each such $t_j$ \emph{supports} $s$ for the dependency $i\in I$ denoted as $\support{i}{s} \dfn \{t_1,\ldots,t_m \}$. 
Clearly, $T\models i$ if and only if $\support{i}{s}\neq\emptyset$ for each $s\in T$.
Moreover, if $\support{i}{s}=\emptyset$ for some $s\in T$ and $i \in  I$, then $s$ can not belong to a repair for $I$.
In the following, we formalize this notion and simulate the semantics for IDs via AFs.

\begin{definition}\label{def:af-inc}
	Let $\calI =\langle T,I \rangle $ be an instance of $\REP$ including a database $T$ and a collection $I$ of IDs.
	Then $\AF{F}{I}$ is the following AF.
	\begin{itemize}
		\item $A\dfn T \cup \{s_i \mid s\in T, i\in  I\}$. 
		\item $R\dfn \{(s_i,s), (s_i,s_i) \mid s\in T, i\in I \}\cup \{(t,s_i)\mid s\in T, i\in  I, t\in \support{i}{s}\}$.
	\end{itemize}
\end{definition}

Intuitively, for each $i\in I$ and tuple $s\in T$, the presence of attacks $(t,s_i)$ for each $t\in \support{i}{s}$ simulates the support relationship between $s$ and tuples in $\support{i}{s}$. 
In other words, each $t\in \support{i}{s}$ attacks $s_i$ and consequently, defends $s$ against $s_i$.
The whole idea captured in this translation is that a tuple $s\in T$ is in a repair for $\mathcal I$ 
if and only if 
for each $i\dfn \inca{x}{y}\in  I$, there is some $t\in T$ such that $s(\tuple x)= t(\tuple y)$
if and only if the argument $s\in A$ is defended against each $s_i$ in $\AF{F}{I}$.


\begin{example}\label{intro:ex-inc}
	Consider $\calI = \langle T,I \rangle $ with database $T=\{s,t,u,v\}$ and IDs $I\dfn \{{\text{Sup\_ID}\subseteq \text{Emp\_ID}}, {\text{Covers\_For}\subseteq \text{Dept}}\}$.
	For brevity, we denote IDs by $I= \{1,2\}$.
	The database and the supporting tuples $\support{i}{w}$ for each $i\in I, w\in T$ are depicted in the table inside Figure~\ref{fig:ex-inc}.
	Informally, a supervisor is also an employee and each employee 
	is assigned a department to cover if that department is short on employees.
	For example, $s(\text{Sup\_ID})=t(\text{Emp\_ID})$, $s(\text{Covers\_For})=t(\text{Dept})=u(\text{Dept})$, and therefore $\support{1}{s}=\{t\}$, $\support{2}{s}=\{t,u\}$.
	Then we have the AF $\AF{F}{I}$ as depicted in Figure \ref{fig:ex-inc}.
	The AF $\AF{F}{I}$ has a unique preferred extension, given by $\{s,t\}$. 
	Clearly, this is also the only repair for $\calI$.
	\begin{figure}
		\centering
		\begin{tikzpicture}[scale=.8, arg/.style={circle,draw=black,fill=white,inner sep=.75mm}]
			\node[arg] (s) at (-.8,2) {$s$};
			\node[arg] (s1) at (-2,2) {$s_1$};
			\node[arg] (s2) at (-1,3) {$s_2$};
			\node[arg] (t) at (-.8,0) {$t$};
			\node[arg] (t1) at (-2,0) {$t_1$};
			\node[arg] (t2) at (-1,-1) {$t_2$};
			\node[arg] (u) at (1.7,1.8) {$u$};
			\node[arg] (u1) at (3,2) {$u_1$};
			\node[arg] (u2) at (2,3) {$u_2$};
			\node[arg] (v) at (1.5,0) {$v$};
			\node[arg] (v1) at (3,0) {$v_1$};
			\node[arg] (v2) at (2,-1) {$v_2$};
		\node (table) at (-1,5.5) {
		\begin{tabular}{l@{\hskip 5pt}|@{\hskip5pt}c@{\;}ccc@{\hskip5pt}}
	$T$ & Emp\_ID  & Sup\_ID & Dept. & Covers\_For\\ \hline
	s 	& JonX1  &  AxeK4 		&  Production 		& Marketing \\
	t 	& AxeK4  &  AxeK4 		&  Marketing  		& Production \\
	u 	& TimX3  &  JonX1 		&  Marketing  		& Distribution \\
	v 	& JonX1  &  AxeK4 		&  Distribution	 	& R\&D \\
		\end{tabular}};
		\node (table) at (6,5.5) {
			\begin{tabular}{cc}
				$S_{1}$ & $S_{2}$\\ \hline
				t & t,u   \\
				t & s  \\
				s,v & v  \\
				t & - \\
			\end{tabular}
		};

			\foreach \f/\t in {s1/s,s2/s,t1/t,t2/t,u1/u,u2/u, v1/v,v2/v}{
				\path[-stealth',draw=blue] (\f) edge (\t);
			}
			\foreach \f/\t in {u1/u1,v1/v1}{
				\path[-stealth',draw=red] (\f) edge[loop right] (\t);
			}
			\foreach \f/\t in {s1/s1,s2/s2,t1/t1,t2/t2,u2/u2,v2/v2}{
				\path[-stealth',draw=red] (\f) edge[loop left] (\t);
			}
			\foreach \f/\t in {t/s1,t/s2,u/s2, t/t1, s/t2, s/u1, v/u1, v/u2, t/v1}{
				\path[-stealth',draw=black] (\f) edge[bend left] (\t);
			}
			
		\end{tikzpicture}\\[-.75em]
		\caption{The AF $\AF{F}{I}$ modelling $\calI$ in Example~\ref{intro:ex-inc}: the red self-loops together with blue arcs depict the attacks for each tuple $s\in T$ due to IDs $i\in I$ and the black arcs model the attacks due to the support set $\support{i}{s}$.}\label{fig:ex-inc}
	\end{figure}
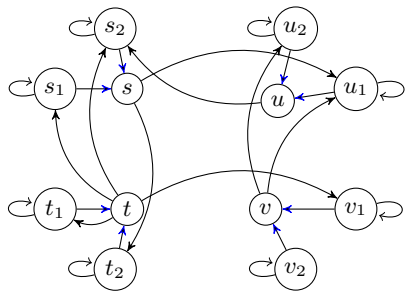
\end{example}

It is worth mentioning that, $\{s,t,u,v\}$ constitutes a naive extension for $\AF{F}{I}$ in Example~\ref{intro:ex-inc}, although this is not a repair for $\mathcal I$.
Clearly, the semantics for IDs in $\AF{F}{I}$ requires admissibility (defending against attacking arguments).
%
We now prove that the repairs for $\calI$ are precisely the preferred extension in $\AF{F}{I}$.

\begin{theorem}\label{thm:ext-inc}
	Let $\mathcal I = \langle T, I \rangle$ be an instance of $\REP$ where $I$ is a set of IDs and let $\AF{F}{I}$ denote the argumentation framework generated by $\mathcal I$.
	Then for every subset $P\subseteq T$, $P \in \repairs(\mathcal I)$ iff $P\in \pref(\AF{F}{I})$.
\end{theorem}
\begin{proof}
	We first prove the reverse direction.
	Let $P\subseteq A$ be a preferred extension in $\AF{F}{I}$, then $P$ must not contain any auxiliary argument $s_i$ corresponding to some ID $i\in I$ since $P$ is conflict-free.
	This implies that $P\subseteq T$, which together with that fact $P$ is admissible (every $s\in P$ is defended against each $s_i\in A$) and maximal under set inclusion yields the proof of the claim.
	
	Conversely, let $P\subseteq T$ denote a repair for $\calI = \langle T, I \rangle$. 
	Then 
	$P$ is conflict-free in $\AF{F}{I}$ since the attacks in $R$ contain at least one argument among the auxiliary arguments ($s_i$) which are not in $P$ (as $P\subseteq T$).
	Moreover, for each $s\in P$ and $i\dfn \inca{x}{y}\in I$, there is some $t\in P$ such that: $s(\tuple x)=t(\tuple y)$.
	This implies that each $s\in P$ is defended against the attack $s_i\in A$. 
	Consequently, $P$ is admissible. 
	To prove that $P$ is also preferred in $\AF{F}{I}$, assume to the contrary that there is an admissible $P'\supset P$ in $\AF{F}{I}$.
	Since $P'$ is also conflict-free, using the same argument as for $P$ we notice that $P'\subseteq T$.
	Now, $P'$ being preferred (together with the claim in reverse direction) 
	implies that $P'$ is a repair for $\calI$ contradicting the fact that $P$ is a subset-maximal repair for $\calI$.
	As a consequence, $P$ is preferred in $\AF{F}{I}$.

	This proves the correctness of our theorem. \qed
\end{proof}

Notice that a framework $\AF{F}{I}$ may not have stable extensions for certain instances $\calI$ including databases $T$ and IDs $I$.
This holds because some arguments can neither be accepted in an extension (e.g., when $\support{i}{s}=\emptyset$ for some $s\in T$ and $i\in I$), nor attacked by arguments in an extension (since arguments in $A$ only attack auxiliary arguments). 
The argument corresponding to the tuple $v$ in Example~\ref{intro:ex-inc} depicts such an argument.
As a result, the stable and preferred extensions do not coincide in general.
Nevertheless, we prove that after a pre-processing, naive, stable and preferred extensions still coincide. 

\paragraph{A pre-processing algorithm for $\AF{F}{I}$.}
Observe that an undefended argument in an AF $\calF$ can not belong to any preferred extension of $\calF$.
The intuition behind pre-processing is to remove such arguments, which are not defended against some of their attacks in $\AF{F}{I}$.
This corresponds to (recursively) removing those tuples in $s\in T$, for which $\support{i}{s}=\emptyset$ for some $i\in I$.
The pre-processing (denoted $\pre(\AF{F}{I})$) applies the following procedure as long as possible. 
\begin{description}
	\item[*] For each $s_i\in A$ such that $s_i$ is not attacked by any $t\neq s_i$: remove $s$ and $s_j$ for each $j\in I$, as well as each attack to and from $s$ and $s_j$. 
\end{description}
We repeat this procedure until convergence. 
Once a fixed point has been reached, the remaining arguments in $A$ are all defended. 
Interestingly, after the pre-processing, removing the arguments with self-loops results in a unique naive extension which is also stable and preferred.
In the following, we also denote by $\pre(\AF{F}{I})$ the AF obtained after applying the pre-processing on $\AF{F}{I}$.
Notice that $\pre$ is basically an adaptation to the AFs of the well-known algorithm for finding a maximal satisfying subteam for inclusion logic formulas~\cite[Lem.~12]{HannulaH22}.

\begin{lemma}\label{lem:pre}
	Let $\mathcal I$ be an instance of $\REP$ and $\AF{F}{I}$ denote the argumentation framework generated by $\mathcal I$.
	Then $\pre(\AF{F}{I})$  can be computed from $\AF{F}{I}$ in polynomial time.
	Moreover, $\pre(\AF{F}{I})$ has a unique naive extension which is also stable and preferred. 
\end{lemma}

\begin{proof}
	The procedure $\pre(\AF{F}{I})$ removes recursively all the arguments corresponding to assignments $s$ such that $\support{i}{s}=\emptyset$ for some $i\in I$.
	Notice that $\support{i}{s}$ can be computed for each $i\in I$ and $s\in T$ in polynomial time.
	Then, $\pre$ stores in a data structure (such as a queue) all the arguments $s$ for which $\support{i}{s}=\emptyset$.
	Finally, each argument $s$ in this queue can be processed turn by turn, adding possibly new arguments when $\pre$ triggers the removal of certain arguments from $A$ and hence from $\support{i}{t}$ for some $t\in A$.
	A fixed-point is reached when every element in the queue has been processed, this gives the size of $A$ as the total number of iterations.  
	Consequently, $\pre$ runs in polynomial time in the size of $\AF{F}{I}$.
	
	Let $\pre(\AF{F}{I}) = (A',R')$ denote the AF generated by the pre-processing.
	To prove the equivalence between extensions, notice that 
	the set of arguments ($S$) without self-attacks in $A'$ form an  
	admissible extension 
	since $\support{i}{s}\neq \emptyset $ for every $s\in S$.
	Further, since $A'\setminus S$ only includes auxiliary arguments, those are all attacked by $S$ and therefore $S$ is stable.
	Finally, $S$ is the only naive extension in the reduced AF since $S$ is the maximal and conflict-free in $\pre(\AF{F}{I})$ and arguments in  $A'\setminus S$ 
	contain self-attacks.
	
	This establishes the correctness of the lemma together with Theorem~\ref{thm:ext-inc}. \qed
\end{proof}

The following observation follows from the proof of Lemma~\ref{lem:pre}.
Intuitively, we can also determine $\somerepair(s,\calI)$ and $\allrepair(s,\calI)$ for each $s\in T$, once the pre-processing has terminated resulting in $\pre(\AF{F}{I})$.
\begin{remark}\label{cor:inc-cs}
	Let $\mathcal I$ be an instance of $\REP$ and $\AF{F}{I}$ denote the argumentation framework generated by $\mathcal I$.
	Then, $\somerepair(s,\calI)$ and $\allrepair(s,\calI)$ is true for every $s\in T$ such that $s\in \pre(\AF{F}{I})$.
\end{remark}
%

\begin{example}[Continue]
	Reconsider the instance $\calI= \langle T, I\rangle $ from Example~\ref{intro:ex-inc}.
	Observe that the argument $v$ is not defended against $v_2$ and therefore can not be in a repair. 
	Then the pre-processing removes $\{v,v_1,v_2\}$ and all the edges to/from arguments in this set.
	This has the consequence that all the arguments which are only defended by $v$ are no longer defended (e.g., $u$).
	Consequently, the arguments $\{u,u_1,u_2\}$ have to be removed as well. 
	After repeating the process for $u$, we notice that no other argument needs to be removed. 
	Hence, the set $\{s,t\}$ is a repair for $\calI$ as well as a $\sigma$-extension in the reduced AF for $\sigma\in\{\naive,\stab,\pref\}$.
\end{example}

\subsection{Simulating Functional and Inclusion Dependencies via AFs}\label{sec:both}

Consider an instance $\calB= \langle T,B \rangle$ with a database $T$ such that $B= D\cup I$ includes functional ($D$) and inclusion ($I$) dependencies.
We apply the pre-processing as a first step, thereby, removing those tuples from $T$ failing some $i\in I$.
In other words, we remove (recursively) all the tuples $s$ from $T$, such that, there is some $i\in I$ with $\support{i}{s}=\emptyset$.
Then, the framework generated by $\calB$ is $\AF{F}{B}\dfn (A,R_{\text D}\cup R_{\text I})$, specified as below.
%
	\begin{itemize}
		\item $A\dfn T \cup \{s_i \mid s\in T, i\in  I\}$. 
		\item $R_{\text D}\dfn \{(s,t),(t,s) \mid \text{ there is some } d\in  D, \text{ s.t. } \{s,t\}\not\models d\}$.
		\item $R_{\text I}\dfn \{(s_i,s), (s_i,s_i) \mid s\in T, i\in I \}\cup \{(t,s_i)\mid s\in T, i\in  I, t\in \support{i}{s}\}$.
		
	\end{itemize}

Interestingly, even if we apply pre-processing, some tuples may not be accepted in combination with each other, as depicted in the following example.

\begin{example}\label{intro:ex-both}
	Consider $\calB = \langle T,B \rangle$ with database $T=\{s,t,u\}$ and atoms $B= D\cup I$ where $D= \{\depas{\text{Sup\_ID}}{\text{Building}}\}$ and $I=\{{\text{Covers\_For}\subseteq \text{Dept}}\}$.
	Moreover, the database $T$ and the support $\support{\incas{\text{Covers\_For}}{\text{Dept}}}{w}$ for each $w\in T$ is depicted in the table inside Figure~\ref{fig:ex-both}.
	Then,  $\{s,t\}\not\models \depas{\text{Sup\_ID}}{\text{Building}}$, and $\{t,u\}\not\models \depas{\text{Sup\_ID}}{\text{Building}}$.
	The resulting AF $\AF{F}{B}$ is shown in Figure \ref{fig:ex-both}, where the edges due to the IDs are depicted in red and blue. 
	Then, the only preferred extensions for $\AF{F}{B}$ is  $\{t\}$.
	Also, the only repair for $\calB$ is $\{t\}$.
	Further, although $\{s,u\}$ is preferred for $\AF{F}{D}$ where $\calD =\langle T,D\rangle$ (ignoring red and blue arcs), and $\{s,t, u\}$ is preferred for $\AF{F}{I}$ where $\calI =\langle T,I\rangle$ (ignoring black arcs), none of them is preferred for $\AF{F}{B}$.
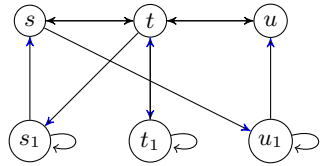
\begin{figure}
	\centering
	\begin{tikzpicture}[scale=.8,arg/.style={circle,draw=black,fill=white,inner sep=.75mm}]
		\node[arg] (s) at (-3,2) {$s$};
			\node[arg] (s1) at (-3,0) {$s_1$};			
		\node[arg] (t) at (-1,2) {$t$};
			\node[arg] (t1) at (-1,0) {$t_1$};
		\node[arg] (u) at (1,2) {$u$};
			\node[arg] (u1) at (1,0) {$u_1$};
	\node (table) at (-3.5,3.5) {
		\begin{tabular}{l@{\hskip 5pt}|@{\hskip5pt}c@{\;}cccc}
	$T$ & Emp\_ID  & Sup\_ID & Dept. & Building & Covers\_for \\
			\hline 
	s & JonX1 & AxeK4 & Production	& B4  & Sales \\
	t & TimX3 & AxeK4 & Sales  			& B2  & Sales \\
	u & AxeK4 & AxeK4 & Marketing  		& B4  & Production\\
		\end{tabular}};
	\node (table) at (3.5,3.5) {
	\begin{tabular}{c}
		$S_i$ \\ \hline
		t \\
		t \\
		s\\
	\end{tabular}
};

	\foreach \f/\t in {s/t,t/s,t/u,u/t}{
			\path[-stealth',draw=black] (\f) edge (\t);
	}
	\foreach \f/\t in {t/s1,t/t1,s/u1,s1/s,t1/t,u1/u}{
		\path[-stealth', draw=blue] (\f) edge (\t);
	}
	\foreach \f/\t in {s1/s1,t1/t1,u1/u1}{
	\path[-stealth',draw=red] (\f) edge[loop right] (\t);
}
	\end{tikzpicture}\\[-.75em]
	\caption{Argumentation framework for modelling dependencies in Example~\ref{intro:ex-both}. Black arcs depict conflicts due to functional, and blue ones due to inclusion dependency}\label{fig:ex-both}
\end{figure}
\end{example}
One consequence of allowing both (functional and inclusion) dependencies is that the preferred and naive extensions do not coincide in general. Moreover, both $\somerepair$ and $\allrepair$ are non-trivial and distinct (cf. Rem.~\ref{cor:dep-cs} and \ref{cor:inc-cs})

\begin{example}[Cont.]
	Reconsider the instance $\calB$ from Example~\ref{intro:ex-both}.
	Then, $\{s,u\}$ is a naive extension for $\AF{F}{B}$ but not preferred.
	Moreover, $t$ is the only tuple for which $\somerepair(t,\calB)$ and $\allrepair(t,\calB)$ is true. 
\end{example}

As the preceding examples demonstrate, in the presence of both types of dependencies, the repairs correspond to somewhat costly (that is, preferred) semantics for AFs.

\begin{theorem}\label{thm:ext-both}
	Let $\mathcal B = \langle T, B \rangle$ be an instance of $\REP$ where $B$ includes FDs and IDs.
	Further, let $\AF{F}{B}$ denote the argumentation framework generated by $\mathcal B$.
	Then for every subset $P\subseteq T$, $P \in \repairs(\mathcal B)$ iff $P \in \pref(\AF{F}{B})$.
\end{theorem}
\begin{proof}
	The correctness follows from the proof of Theorem~\ref{thm:ext-dep} and~\ref{thm:ext-inc}.
	The conflict-freeness and admissibility of $P$ implies that each FD and ID in $B$, respectively, is true in $P$.
	The converse follows the same line of argument. 
	Finally, the maximality of repairs in $\calB$ corresponds to the maximality of extensions in $\AF{F}{B}$.\qed
\end{proof}

The existence of a non-empty extension and the credulous reasoning for preferred semantics are both $\NP$-complete, 
and skeptical reasoning is even harder and $\PiP$-complete. 
Next we establish that when both types of dependencies are allowed, $\REP, \somerepair$ and $\allrepair$ actually have the same complexity as, respectively,  
the existence, credulous and skeptical reasoning for preferred semantics.

\begin{theorem}\label{thm:cred-both}
	Let $\mathcal B = \langle T, B \rangle$ be an instance of $\REP$ where $B$ is a set of FDs and IDs.
	Further, let $s\in T$ be a tuple.
	Then, the problems 
	$\REP$ and $\somerepair(s,\calB)$ are both 
	\NP-complete.
\end{theorem}

\begin{proof}
	The membership is easy in both cases. 
	Given $P\subseteq T$ such that $s\in P$ (resp., $P$ is non-empty), one can decide in polynomial time whether $P\models B$. 
	Notice that we do not need to check the maximality, since if there is a repair (non-empty) for $\calB$ containing $s$ then there is also a maximal (non-empty) repair containing $s$.  
	
	For hardness, we first reduce from $\SAT$ to $\somerepair(s,\calB)$. 
	Towards the end, we highlight the required changes to reduce $\SAT$ to $\REP$.
	Let $\varphi \dfn \{C_i\mid {i\leq m}\}$ be a formula over propositions $\{p_1,\ldots, p_n\}$.
	Then, we construct a database $T$ and a collection $B$ of FDs and IDs over a set $V \dfn \{t_i,u_i \mid 0\leq i \leq m\}$ of attributes.
	Our encoding works as follows.
	\begin{itemize}
		\item $B$ contains a single FD $\depas{t_0}{u_0}$ to encode that each proposition takes at most one value in $\{0,1\}$. Moreover, $B$ contains $m$ inclusion dependencies $\incas{t_i}{u_i}$ for each $1\leq i\leq m$ to assure that each clause $C_i$ is satisfied.
		\item  The database $T\dfn \{s_\varphi\}\cup\{s_j,  \bar s_j \mid 1\leq j\leq n\}$ is constructed in such a way that: (C1) the pair $s_j, \bar s_j$ fails the FD $\depas{t_0}{u_0}$ for each $1\leq j\leq n$, thereby ensuring that any repair contains at most one tuple from $\{s_j, \bar s_j\}$, and (C2) for each ID  $\incas{t_i}{u_i}$ the value of $s_\varphi(t_i) $ is shared only by the tuple $s\in T\setminus\{s_\varphi\}$ such that their corresponding literal satisfies the clause $C_i$.
		Formally, we let $\dom(T)\dfn\{c_i \mid i\leq m\}\cup\{p_1,\ldots, p_n\}\cup\{0,1\}$.
		Then, (C1) is achieved by setting $s_j(t_0)=\bar s_j(t_0)=p_j$, $s_j(u_0)=1$, and $\bar s_j(u_0)=0$. Moreover, we also let $s_\varphi(t_0)= s_\varphi(u_0)=0$. 
		To achieve (C2) we let $s_\varphi(t_i)=c_i$ for each $1\leq i\leq m$ and $s_\varphi(u_i)=0$. Then, we let $s_j(t_i)=s_j(u_i)=c_i$ if $p_j\in C_i$, $\bar s_j(t_i)=\bar s_j(u_i)=c_i$ if $\neg p_j\in C_i$, and we let $s_j(t_i)=s_j(u_i) = 0 = \bar s_j(t_i)= \bar s_j(u_i)$ in the remaining cases. 
	\end{itemize}
	Clearly, $T\not\models B$ due to the presence of the pair $s_j, \bar s_j$ and the FD $\depas{t_0}{u_0}$.
	Then, a repair $P\subseteq T$ for $\depas{t_0}{u_0}$ contains exactly one tuple from each such pair. 
	Finally, $P\models \incas{t_i}{u_i}$ for ${1\leq i\leq m}$ iff $P$ contains at least one tuple $s\in \{s_j,\bar s_j\}$ corresponding to $x\in \{p_j,\neg p_j\}$ such that $x\in C_i$ for each $C_i\in\varphi$.
	This completes the proof since $\varphi$ is satisfiable if and only if $\somerepair({s_\varphi},\calB)$ is true. 
	
	To reduce $\SAT$ into $\REP$, we use an additional ID, $\incas{t_{m+1}}{u_{m+1}}$ and modify $T$ in such a way that every repair $P$ for $\calB$ necessarily contains $s_\varphi$, thereby proving the equivalence as before. 
	This is achieved by adding a new element $c_{m+1}\in\dom(T)$ and setting $s_{\varphi}(t_{m+1}) = s_{\varphi}(u_{m+1})= c_{m+1}$, as well as $s_j(t_{m+1}) = \bar s_j(t_{m+1}) = c_{m+1} $ and $s_j(u_{m+1}) = \bar s_j(u_{m+1}) = 0 $ for each $1\leq j\leq n$.
	This has the effect that one can not construct a subset-repair for $\calB$ by excluding $s_\varphi$ and therefore, there is a non-empty repair for $\calB$ iff $\varphi$ is satisfiable.
	This completes the proof for both cases.
	\qed 
\end{proof}

We provide an example for better understanding of the reductions from the proof of Theorem~\ref{thm:cred-both}.

\begin{example}\label{ex:red-both}
	Let $\varphi \dfn \{x \lor y, \neg x \lor \neg y, \neg x \lor y\}$ be a propositional formula.
	Then, our reduction for $\somerepair$ yields a database $T\dfn \{s_\varphi, s_x,s_y,\bar s_x, \bar s_y\}$ and a collection $B\dfn \{\depas{t_0}{u_0}\}\cup \{\incas{t_i}{u_i} \mid 1 \leq i \leq 3\}$ of dependencies as depicted in the left side of Table~\ref{tab:red-both}.
	Notice that the only satisfying assignment for $\varphi$ is given by $\{x\mapsto 0, y\mapsto 1\}$, corresponding to the  repair $\{s_\varphi, \bar s_x, s_y\}$ for $\calB$ containing $s_\varphi$, and consequently $\somerepair({s_\varphi},\calB)$ is true.
	
	Moreover, although each of $\{s_x,s_y\},\{s_x,\bar s_y\},\{\bar s_x,\bar s_y\}$ is also a repair for $\calB$, 
	none of them contains $s_\varphi$.
	Then, the second part of our reduction (for $\REP$) adds the inclusion dependency $\incas{t_{4}}{u_{4}}$ to $B$ and expands the database $T$ by two columns in the right side of Table~\ref{tab:red-both}.
	This results in $\{s_\varphi,\bar s_x, s_y\}$ being the only $\REP$ for $\calB$ corresponding to the satisfying assignment for $\varphi$.
\begin{table}[ht]
	\centering
	\begin{tabular}{l@{\hskip5pt} |@{\hskip5pt} cc @{\hskip5pt}| @{\hskip5pt}cccccc}
		& $t_0$ & $u_0$  & $t_1$ & $u_1$& $t_2$ & $u_2$ & $t_3$ & $u_3 $\\ \hline 
	$s_\varphi$  
	& 0 & 0  & $c_1$ & 0 & $c_2$ & 0 & $c_3$ & $0$\\ \hline
	$s_x$ 
	& $x$ & 1  & $c_1$ & $c_1$ & $0$ & 0 & $0$ & $0$\\
	$\bar s_x$ 
	& $x$ & 0  & $0$ & 0 & $c_2$ & $c_2$ & $c_3$ & $c_3$\\ \hline 
	$s_y$ 
	& $y$ & 1  & $c_1$ & $c_1$ & $0$ & 0 & $c_3$ & $c_3$\\
	$\bar s_y$
	& $y$ & 0  & $0$ & 0 & $c_2$ & $c_2$ & $0$ & $0$
	\end{tabular}
\hspace{1cm}
	\begin{tabular}{@{\hskip2pt}|@{\hskip5pt} cc @{\hskip5pt}}
	 $t_4$ & $u_4$ \\ \hline 
	 $c_4$ & $c_4$  \\ \hline
	 $c_4$ & $0$ \\
	 $c_4$ & 0 \\ \hline 
	 $c_4$ & $0$ \\
	 $c_4$ & $0$ 
\end{tabular}
\caption{The database corresponding to the formula $\varphi$ from Example~\ref{ex:red-both}}
\label{tab:red-both}
	\end{table}
\end{example}
%

Next, we prove that $\allrepair(s,\calB)$ is even harder and $\PiP$-complete.

\begin{theorem}\label{thm:skep-both}
	Let $\mathcal B = \langle T, B \rangle$ be an instance of $\REP$ including a set $B$ of FDs and IDs.
	Further, let $s\in T$ be a tuple.
	Then, the problem 
	$\allrepair(s,\calB)$ 
is \PiP-complete.
\end{theorem}

\begin{proof}
	
	For membership, one can guess a subset $P\subseteq T$ as a counter example for $s$, that is, $s\not \in P$ and $P$ is a subset-repair for $\calB$, which can be decided in polynomial time. 
	However, to determine whether $P$ is maximal, one has to use oracle calls for guessing subsets $P' \subseteq T$, with $s\not\in P'$ to determine whether $P'\models B$ and $P'\supset P$. 
	This gives an upper bound of $\co\NP^{\NP}$ (equivalently $\PiP$).

	For hardness, we use a similar idea as in the proof of Theorem~\ref{thm:cred-both} and reduce from an instance $\Phi$ of the $\PiP$-complete problem $\tqbf$, where $\Phi = \forall Y \exists Z \varphi(Y,Z)$ and $\varphi \dfn \{C_i \mid i\leq m\}$ is a $\CNF$.
	We let $X=Y\cup Z$ and construct an instance $\calB \dfn \langle T,B \rangle$ over a set $V \dfn \{t_i,u_i \mid 0\leq i \leq m+1\}$ of attributes.
	As in the proof of Theorem~\ref{thm:cred-both}, $B$ contains a FD $\depas{t_0}{u_0}$ and a collection of IDs $I\dfn \{ \incas{t_i}{u_i} \mid 1\leq i\leq m+1\}$ to encode whether each clause $C_i\in \varphi$ is satisfied. 
	Moreover, the additional ID $\incas{t_{m+1}}{u_{m+1}}$ encodes the existentially quantified variables $Z$.
	The database is also constructed as in the proof of Theorem~\ref{thm:cred-both}, except for the attributes $\{t_{m+1}, u_{m+1}\}$.
	These attributes encode the effect that $s_\varphi$ \emph{supports} variables $\{z,\bar z\}$ for each $z\in Z$ via the inclusion dependency $\incas{t_{m+1}}{u_{m+1}}$.
	This is achieved by letting $s_z(t_{m+1})=\bar s_z(t_{m+1})= c_{m+1}$ for each such $z \in Z$, as well as $s_z(u_{m+1})=\bar s_z(u_{m+1})= 0$ and $s_\varphi(t_{m+1}) = s_\varphi(u_{m+1}) = c_{m+1}$.

	For correctness, notice that every interpretation $I_Y$ over $Y$ (seen as a subset of $Y$) corresponds to a subset $P_Y =\{s_y \mid y\in I_Y\} \cup \{\bar s_y \mid y \not \in I_Y\}$.
	Then, $P_Y\models B$:  $\depas{t_0}{u_0}$ is true since $P_Y$ includes only one of $s_y, \bar s_y$ and each  $\incas{t_i}{u_i}$ is true since $s(t_i)=s(u_i)$ for each $s\in T\setminus \{s_\varphi\}$ and $i\leq m+1$.
	Moreover, we have that $s_\varphi \not \in P_Y$.
	Now, in order to extend $P_Y$ by adding $s_z$ or $\bar s_z$ for any $z\in Z$, $s_\varphi$ must be added as well due to the ID $\incas{t_{m+1}}{u_{m+1}}$.
	However, in order to add $s_\varphi$, we have to find $I_Y$ and $I_Z$ that together satisfy $\varphi$ due to the IDs $\incas{t_i}{u_i}$ for $i\leq m$.
	As a result, for any interpretation $I_Y$ over $Y$: $I_Y \cup I_Z \models \varphi$ if and only if $P_Y$ is not a repair for $\calB$ (since, $P_Y\cup P_Z\cup\{s_\varphi\}$ is a repair in such a case).
	Equivalently, there is a repair for $\calB$ not containing $s_\varphi$ if and only if the formula $\Phi$ is false.
	As a consequence, $\Phi$ is true if and only if every repair for $\calB$ contains $s_\varphi$ if and only if $\allrepair(s_\varphi, \calB)$ is true.
	\qed 
\end{proof}

We provide an example for better understanding of the reduction from the proof of Theorem~\ref{thm:skep-both}.

\begin{example}\label{ex2:red-both}
	Let $\Phi = \forall y_1 y_2 \exists z_3z_4((y_1 \lor y_2 \lor z_3) \land (y_2 \lor \neg z_3 \lor \neg z_4) \land (y_2 \lor z_3 \lor z_4))$ be a $\tqbf$.
	Then, our reduction yields a database $T\dfn \{s_\varphi, s_1,s_2,s_3,s_4,\bar s_1, \bar s_2,\bar s_3,\bar s_4\}$ and a collection $B\dfn \{\depas{t_0}{u_0}\}\cup \{\incas{t_i}{u_i} \mid i\leq 4\}$ of dependencies as depicted in Table~\ref{tab2:red-both}.
	The reader can verify that the formula $\Phi$ is true and that $\allrepair(s_\varphi,\calB)$ is true as well.
	\begin{table}[t]
		\centering
		\begin{tabular}{l@{\hskip5pt}|@{\hskip5pt} cc@{\hskip5pt}|@{\hskip3pt}cccccccc}
			& $t_0$ & $u_0$  & $t_1$ & $u_1$& $t_2$ & $u_2$ & $t_3$ & $u_3 $ & $t_4$ & $u_4$\\ \hline 
			$s_\varphi$  
			& 0 & 0  & $c_1$ & 0 & $c_2$ & 0 & $c_3$ & $0$ & $c_4$ & $c_4$ \\ \hline
			$s_1$ 
			& $y_1$ & 1  & $c_1$ & $c_1$ & $0$ & $0$ & $0$ & $0$ & 0 & $0$\\
			$\bar s_1$ 
			& $y_1$ & 0  & $0$ & 0 & $0$ & $0$ & $0$ & $0$ & 0 & $0$\\ \hline 
			$s_2$
			& $y_2$ & 1  & $c_1$ & $c_1$ & $c_2$ & $c_2$ & $c_3$ & $c_3$ & 0 & $0$\\
			$\bar s_2$
			& $y_2$ & 0  & $0$ & 0 & $0$ & $0$ & $0$ & $0$ & 0 & $0$\\ \hline 
			$s_3$ 
			& $z_3$ & 1  & $c_1$ & $c_1$ & $0$ & 0 & $c_3$ & $c_3$ & $c_4$ & $0$ \\
			$\bar s_3$
			& $z_3$ & 0  & $0$ & 0 & $c_2$ & $c_2$ & $0$ & $0$ & $c_4$ & $0$ \\ \hline
			$s_4$ 
			& $z_4$ & 1  & $0$ & $0$ & $0$ & 0 & $c_3$ & $c_3$ & $c_4$ & $0$\\
			$\bar s_4$
			& $z_4$ & 0  & $0$ & 0 & $c_2$ & $c_2$ & $0$ & $0$ & $c_4$ & $0$
		\end{tabular}
		\caption{The team corresponding to the instance $\calB$ from Example~\ref{ex2:red-both}}
		\label{tab2:red-both}
	\end{table}
	
\end{example}

We conclude this section by noting that the least requirements for subset-repairs in the presence of both atoms is conflict-freeness and admissibility.
Furthermore, although admissible extensions for $\AF{F}{B}$ yield subset-repairs for $\calB$, those are not guaranteed to be maximal.
%
%
\section{From Unirelational to Multirelational Databases}\label{sec:multi}

A database instance in multirelational setting consists of a collection $\calT=(T_1,\ldots T_m)$, where each $T_j$ is a database corresponding to a relational schema $T_j(x_{j,1} ,\ldots, x_{j,n_j})$ of arity $n_j$.
As before, $T_j$ denotes the relation name and $x_{j,1},\ldots,x_{j,n_j}$ are distinct {attributes}.
In the following, we let $T= \bigcup_{j\leq m} T_j$ denote the set of all the tuples in $\calT$.
A {functional dependency} over $\calT$, and the satisfaction for FDs is defined as before, i.e., an expression of the form $\depa{x}{y}$ for sequences $\tuple x, \tuple y$ of attributes in $T_j$ for some $j\leq m$.
However, an {inclusion dependency} may address attributes from two different tables.
That is, $i= \inca{x}{y}$ is an ID between $T_j$ and $T_k$ if $\tuple{x}$ and $\tuple y$ are sequences of attributes over $T_j$ and $T_k$, respectively.
We call $T_j$ the \emph{source} and $T_k$ the \emph{target} of $i$, denoted as $\source(i)$ and $\target(i)$. 
Then, $\calT\models \inca{x}{y}$ if for each $s\in \source(\inca{x}{y})$, there is some $t\in \target(\inca{x}{y})$ such that $s(\tuple x)=t(\tuple y)$.
By slightly abusing the notation, if $d\dfn \depa{x}{y}$ is an FD over attributes in $T_j$, then we write $\source(d)=T_j$.
Finally, we define the notion of (subset-maximal) repairs similar to the case of unirelational databases,  i.e., $\mathcal P = (P_1,\ldots,P_m)$ where $P_j\subseteq T_j$ 
for $j\leq m$ such that $\mathcal P$ satisfies each dependency in $ B$.

The construction from Section~\ref{sec:rep-afs} for unirelational setting can be expanded to allow databases with more than one relations.
The encoding for FDs remains the same as before, whereas for IDs (between $T_j$ and $T_k$), we create auxiliary arguments only for tuples in $\source(i)$, which can be attacked by arguments corresponding to tuples in the $\target(i)$.
As before, we denote by $\support{i}{s}=\{t \mid t\in\target(i),s(\tuple x)=t(\tuple y)\}$ the tuples supporting $s$ for an ID $i=\inca{x}{y}$.
Given an instance $\calB = \langle \calT, B\rangle $ of a multirelational database $\calT$ and a collection $B = D\cup I$ of FDs $D$ and IDs $I$, we construct the AF~$\AF{F}{B}=(A,R_\text{D}\cup R_\text{I})$ as follows.
\begin{itemize}
	\item $A \dfn  \{s \mid s\in T \} \cup \{s_i \mid s\in \source(i) \text{ for } i\in I\}$,
	\item $R_\text{D} \dfn \{(s,t),(t,s) \mid s,t \in \source(d) \text{ for } d \in D  \text{ and } \{s,t\}\not\models d \}$,
	\item $R_{\text{I}}\dfn \{(s_i,s), (s_i,s_i) \mid s\in \source(i) \text{ for } i\in I \}\cup \{(t,s_i)\mid t\in \support{i}{s} \text{ for } i\in I\}$.
\end{itemize}

Then, a similar argument as in the proof of Theorem~\ref{thm:ext-both} allows us to establish that repairs for an instance $\calB =\langle\calT,B\rangle $ including a multirelational database $\calT$ and a collection $B$ of FDs and IDs are precisely the preferred extensions in $\AF{F}{B}$.

\section{Concluding Remarks}
\paragraph{Overview.}
We 
simulated the problem of finding repairs of an inconsistent database under functional and inclusion dependencies by Dung's argumentation frameworks.
Our main results (See Table~\ref{table:cont}) indicate that subset maximal repairs correspond to naive extensions when only one type of dependency is allowed, whereas only preferred extensions yield all the repairs when both FDs and IDs are allowed.
Further, for the problem to determine whether a tuple is in some (resp., every) repair, we establish the same complexity bounds as the complexity of credulous (skeptical) reasoning for preferred-semantics in AFs.

\paragraph{Discussion and Future Work.}
We would like to point out that, although, the conflict relation in the presence of functional dependencies and a connection between preferred extensions and subset maximal repairs is known for FDs~\cite{BienvenuB20}, the main contributions of our work establishes the relation between extensions of AFs when inclusion dependencies are also allowed.
This novel contribution opens up several directions for future work.
First and foremost, the authors believe that the connection between repairs in the setting of inconsistent databases and extensions in AFs is stronger than what is established here.
Intuitively, one can model the attack relationship via functional, and defense/support via inclusion dependencies.
A precise formulation of this transformation will allow us to simulate AFs via inconsistent databases by considering FDs and IDs.
However, this intuition needs further exploration and is therefore left for future work. 

Further notable future work may consider whether the idea presented here can be generalized to other well-known types of tuple or equality generating dependencies (also known as \emph{tgds} and \emph{egds}).
%
%
Moreover, we would like to explore whether the consistent query answering (CQA) under inconsistency-tolerant semantics can also be tackled via the argumentation approach.
Then, one can consider incorporating the information about priorities among tuples into the resulting AFs, that is, extending the translations presented in this work to the setting of prioritized repairing and consistent query answering~\cite{FaginKK15,kimelfeld-17,KimelfeldLP20}.

Another promising direction to consider next is the exploration of an explainability dimension. 
Given an instance $\calB$ including a database $T$ and a collection $B$ of dependencies, then the proposed AF $\AF{F}{B}$ lets one determine the \emph{causes} why some tuples are not in some repair (or all repairs).
For IDs, the auxiliary arguments modelling each dependency in $B$ can serve this purpose.
For FDs, we believe that annotating arguments (or the attack relation between a pair of arguments) by the FDs involved in the conflict can achieve the goal.
As a result, one can look at the AF $\AF{F}{B}$ and read from it the FDs or IDs which a tuple $s$ failing $\somerepair(s,\calB)$ or $\allrepair(s,\calB)$ participates in.
Then, subsets of the atoms and/or possibly tuples in a database can be considered as \emph{explanations}.
Such explanations seem interesting in modeling scenarios where the data (database) has higher confidence than the depedencies; for example, if dependencies are mined over some part of the existing data.
An explanation then informs that the data (and hence tuples therein) should be kept, whereas dependencies need to be screened and further analyzed. 
\bibliographystyle{splncs04}
\bibliography{main}

\end{document}